\documentclass[copyright,creativecommons]{eptcs}
\providecommand{\event}{LFMTP 2025}
\usepackage{longtable,booktabs}
\usepackage{mathtools}
\usepackage{mathrsfs}
\usepackage{graphicx}
\usepackage{amsmath,url,hyperref}
\usepackage{amsfonts}
\usepackage{amssymb}
\usepackage{xifthen}
\usepackage[capitalise]{cleveref}
\usepackage{paralist}
\usepackage{enumitem}
\usepackage[utf8]{inputenc} 
\usepackage{xcolor}    
\usepackage{amsthm}
\usepackage{bussproofs}
\usepackage{subcaption}
\usepackage{doi}
\newtheorem{definition}{Definition}
\newtheorem{lemma}{Lemma}

\newcounter{ruleName}
\newcommand{\ruleNameLbl}[1]{\refstepcounter{ruleName}(#1)\label{rule:#1}}
\newcounter{s}
\newcommand{\atomsort}[1]{\ensuremath{\mathsf{#1}}}
\newcommand{\atomset}[2]{\ensuremath{{\mathbb{#1}}_{#2}}}
\newcommand{\atom}[1]{\ensuremath{#1}}

\newcommand{\datasort}[1]{\ensuremath{\mathrm{#1}}}
\newcommand{\absort}[2]{\ensuremath{\ll\!\!#1\!\!\gg #2}}
\newcommand{\absterm}[2]{\ensuremath{\ll\!\!#1\!\!\gg #2}}
\newcommand{\dabsort}[3]{\ensuremath{\ll\!\!\ntyping{#1}{#2}\!\!\gg #3}}
\newcommand{\dabsterm}[3]{\ensuremath{\ll\!\!\ntyping{#1}{#2}\!\!\gg #3}}
\newcommand{\concr}[2]{#1[#2]}
\newcommand{\concrVect}[2]{#1\vect{[#2]}}

\newcommand{\vect}[1]{\ensuremath{\overline{#1}}}
\newcommand{\simpleconstrdecl}[3]{\ensuremath{#1 : #2 \rightarrow\ #3}}

\newcommand{\depconstrdecl}[3]{\ensuremath{#1 : (#2) \rightarrow\ #3}}
\newcommand{\sconstr}[2]{\ensuremath{#1\,#2}}
\newcommand{\typing}[2]{\ensuremath{#1 : #2}}
\newcommand{\ntyping}[2]{\ensuremath{#1{:}#2}}
\newcommand{\fresh}[2]{#1 \# #2}

\newcommand{\sort}{\ensuremath{\mathsf{sort}}}
\newcommand{\data}{\ensuremath{\mathsf{data}}}

\newcommand{\sortconstrdecl}[2]{\ensuremath{#1 : #2 \rightarrow\ \data}}
\newcommand{\fconstrdecl}[4]{\ensuremath{#1 : #2 \rightarrow\ #3\ ;\ #4}}

\newcommand{\typeAndFresh}[3]{#1\colon #2\,;\,#3}
\newlength{\jdgsep}
\newcommand{\emptyTel}{\cdot}
\newcommand{\consTel}[3]{#1,(\typing{#2}{#3})}
\newcommand{\emptyCtx}{\cdot}
\newcommand{\consCtx}[3]{#1,(\typing{#2}{#3})}

\newcommand{\lookupConstSig}[3]{#2(#1) = #3}

\newcommand{\sigjdgmnt}[1]{\ensuremath{\vdash #1\hspace{\jdgsep} \text{sig-ok}}}
\newcommand{\teljdgmntsig}[3]{\ensuremath{#1\vdash_{#2} #3 \hspace{\jdgsep} \text{tel-ok}}}
\newcommand{\telfitsjdgmntsig}[5]{\ensuremath{#1;#2\vdash_{#3} #4 \text{ fits } #5}}
\newcommand{\sortjdgmntsig}[5]{\ensuremath{#1;#2\vdash_{#4}  #5 \ \ \sort}}

\newcommand{\clsortjdgmntsig}[4]{\ensuremath{#1;\emptyCtx\vdash_{#3} #4 \ \ \sort}}
\newcommand{\ctxjdgmntsig}[4]{\ensuremath{#1\vdash_{#3} #4 \hspace{\jdgsep} \text{ctx-ok}}}

\newcommand{\tyjdgmntsig}[6]{\ensuremath{#1;#2\vdash_{#4} \typing{#5}{#6}}}
\newcommand{\tyjdgmnt}[5]{\ensuremath{#1;#2\vdash \typing{#4}{#5}}}
\newcommand{\tyjdgmnts}[2]{\ensuremath{\vdash \typing{#1}{#2}}}

\newcommand{\SigEmpty}{\ensuremath{\langle \rangle}}
\newcommand{\SigConsSort}[4]{\ensuremath{#1 , \langle %
    \typeAndFresh{#2}{#3\to\data}{#4} \rangle}}
\newcommand{\SigConsFun}[5]{\ensuremath{#1 , \langle %
    \typeAndFresh{#2}{#3\to #4}{#5} \rangle}}

\newcommand{\cnst}[1]{\textsl{#1}}

\newcommand{\form}{\texttt{Form}}
\newcommand{\prf}{\ensuremath{\mathcal{D}}}
\newcommand{\impl}{\ensuremath{\supset}}

\newcommand{\name}{\textsl{V}}
\newcommand{\ofType}{\colon}
\newcommand {\deffont}  [1] {\emph{#1}}
\newcommand {\id}       [0] {\mathsf{id}}
\newcommand {\act}      [0] {\cdot}

\newcommand {\prmAct}   [2] {#1 \act #2}
\newcommand{\aSub}[3]{\ensuremath{[#1 \mapsto #2]#3}}

\newcommand {\aSubAct} [2] {#1 \, #2}

\newcommand{\varTrmD} [3] {\ensuremath{#3}} 

\newcommand {\absTrmD}  [3] {\absterm{#1 \ofType #2} #3}
\newcommand {\appTrm}  [2] {\sconstr{#1}{#2}}
\newcommand {\absTypD}  [3] {\absort{#1 \ofType #2}{#3}}

\newcommand {\susp}    [2] {\ifthenelse{\isempty{#1}}
                                       {\ifthenelse{\isempty{#2}}{}{#2}}
                                       {\ifthenelse{\isempty{#2}}{#1}{#1  | #2}}}
\newcommand {\prmCat}   [2] {#1 \mathrel{\circ} #2}

\newcommand {\cent}     [0] {\vdash}
\newcommand {\hsh}      [0] {\#}
\newcommand {\frsh}    [0] {\mathrel{\hsh}}

\newcommand {\dom}      [1] {\mathsf{dom}(#1)}

\newcommand {\trans}    [2] {(#1 \ #2)}

\newcommand {\al}       [0] {{\approx_\alpha}}
\newcommand {\aleq}     [0] {\mathrel{\al}}

\newcommand {\joinTsubs} [2] {#1 \circ #2}

\newcommand{\frshRuleName}[1]{\ensuremath{(\mathit{#1})^\hsh}}
\newcommand{\alphaRuleName}[1]{\ensuremath{(\mathit{#1})^\alpha}}

\newcommand{\concat}{+\kern-0.8em+}

\newcommand{\Trm}{\texttt{Term}}
\newcommand{\Frm}{\texttt{Form}}
\newcommand{\enc}{\texttt{enc}}
\newcommand{\dec}{\texttt{dec}}
\newcommand{\deH}{\ensuremath{\delta\enc_H}}
\newcommand{\dta}{\texttt{data}}
\newcommand{\LATy}{\ensuremath{\Lambda^{\rightarrow}\texttt{Type}}}

\newcommand{\LTrm}{\ensuremath{\Lambda\texttt{Term}}}

\begin{document}
\title{Dependently Sorted Nominal Signatures}
\author{Maribel Fernández\institute{King's College London, UK} 
\and Miguel Pagano\institute{Univ. Nac. Córdoba, Argentina}
\and Nora Szasz \qquad Álvaro Tasistro\institute{Universidad ORT Uruguay}
}

\def\titlerunning{}
\def\authorrunning{M. Fernández, M. Pagano, N. Szasz, A. Tasistro}
\def\copyrightholders{Fernández, Pagano, Szasz, Tasistro}

\maketitle
\begin{abstract}
  We investigate an extension of 
  nominal many-sorted signatures in which abstraction has a form of instantiation, called generalised concretion, as elimination operator (similarly to lambda-calculi). 
  Expressions are then classified using a system of sorts and sort families that respects alpha-conversion (similarly to dependently-typed lambda-calculi) but
  not allowing names to carry abstraction sorts, thus constituting a first-order dependent sort system. 
  The system can represent forms of judgement and rules of inference of several interesting calculi. 
  We present rules and properties of the system as well as experiments of representation, and discuss how it constitutes a basis on which to build a type theory where raw expressions with alpha-equivalence are given a completely formal treatment.

\noindent
  Keywords: {Nominal Terms; Logical Frameworks; Dependent Types.}
\end{abstract}
\section{Introduction}

We present a generalisation of Pitts' nominal many-sorted signatures~\cite{pitts:alpsri}, where sorts can now depend on terms, yielding a \emph{dependently sorted system} that inherits the distinctive first-order algebraic flavour of nominal signatures.
We show that this system can serve as a basis for a logical framework.

Nominal logic programming languages such as $\alpha$-Prolog~\cite{CheneyJ:alpp} provide support for the specification of data structures that include bound names and for the formalisation of their properties. Resolution can be used to prove properties. However, there is no type distinction between solvable/unsolvable goals.
Gabbay et al~\cite{gabbay:capasn-jv,gabbay:oneaah-jv,gabbay:curhid} proposed alternative formulations of nominal systems with meta-variables that can be used to represent schematic proofs.
These systems, termed one-and-a-halfth-level or two-level calculi, have type systems closer to simply typed lambda calculus.
So, for example, one can introduce signatures for First-Order Logic (FOL) and the type system will ensure that the equality of two terms is well-typed only when predicated on terms of the sort corresponding to terms of the object language. On the other hand, again there is no type distinction between theorems and contradictions. 
To express that $\phi$ is a theorem, one has to construct not a term, but a derivation of $\phi = \top$ from the axioms of the theory. Notice that this equality corresponds to a different form of judgement. Besides the ample evidence of formalisation of mathematics and computer science in Higher-Order Logic (HOL) and Isabelle/HOL, Bordg, Paulson and Li~\cite{DBLP:journals/em/BordgPL22} have shown that 
sophisticated mathematical constructions can be formalised in simply typed lambda calculus.

With dependent types one can introduce the type (or \emph{sort} as we use in this article) $\form$ of well-formed formulas and a family of sorts $\prf{(\phi)}_{\{\phi \in \form\}}$ of proofs of theorems. Now both well-formedness and theoremhood of the object language are represented  as typing judgements in the meta-language.

In this paper we aim to demonstrate that a quite simple nominal language with dependent sorts can be used as a logical framework \cite{harper:fradl}.
The nominal foundation introduces a theory of expressions subject to alpha-conversion,
on top of which a system of dependent sorts is built that respects  alpha-equivalence. To this effect, we let names carry sorts of data ---but not abstraction (or "higher-order") sorts. This yields a limited form of computation associated to the elimination of abstractions (concretion), that can be solved at the level of syntactic meta-definitions.
Thus the language becomes a first-order dependent sorts system.
We have developed several examples showing its ability to represent binding structures: First Order Logic, Lambda Calculi and Higher-Order Logic. In this paper we show some of these, all treated employing a shallow-encoding strategy in which the object-language substitution is directly implemented with our concretion operation. 
In addition, a deep encoding of the pure Lambda Calculus is shown to allow the  formulation of an alpha-structural induction principle as described and justified in \cite{pitts:alpsri}. We deem this example as opening a path for extending the present system with full rules of computation (i.e. possibly recursive definitions) respecting alpha-conversion, thus yielding a dependent type theory \`a la Martin-L\"of with a nominal syntactic foundation providing the treatment of binding at the infrastructural level.

In \cref{section:related-work} we briefly discuss  related work, highlighting the systems that are closer to our approach and comparing with previous works on dependently sorted first-order signatures~\cite{Cartmell86,SterlingJ:algttu} as well as dependent type systems based on extensions of the $\lambda$-calculus with nominal features~\cite{CheneyJ:depntt,PittsA:deptta}.

\section{Preliminaries}

\subsection{Simple Nominal Signatures}

We start with a review of  nominal signatures~\cite{UrbanC:nomu-jv,pitts:alpsri}. 
Here and in the rest of the article we overline symbols or phrases to denote sequences of the species corresponding to the overlined entity.

A signature is a triple $(\mathcal{S}, \atomset{A}, \Sigma)$, where $\mathcal{S}$ is a set of \emph{basic sorts} including data sorts and name sorts; \atomset{A}{} is a family of countable infinite sets of names (\emph{atoms}), indexed by name sorts. Given $\mathcal{S}$ and $\atomset{A}{}$ the following grammar generates \emph{the sorts of terms}:
\[
  \gamma ::=\  \atomsort{a}\ |\ \datasort{s}\ |\ \absort{\atomsort{a}}{\vect{\gamma}}
\]
In this grammar $\atomsort{a}$ is a \emph{name} sort, $\datasort{s}$ is a \emph{data} sort, and $\absort{\atomsort{a}}{\vect{\gamma}}$ is an \emph{abstraction} sort. $\Sigma$ is a set of declarations of  (term-)\emph{constructors}, each with its \emph{arity}, i.e. $\simpleconstrdecl{\kappa}{\vect{\gamma}}{\datasort{s}}$.

The set of \emph{sorted} terms for the signature $(\mathcal{S}, \atomset{A}, \Sigma)$ is given by the following \textit{sorting} rules:\\

{\centering
    \AxiomC{}
    \LeftLabel{(atom)}
    \RightLabel{$\atom{a} \in \atomset{A}{\atomsort{a}}$}
    \bottomAlignProof
    \UnaryInfC{\typing{\atom{a}}{\atomsort{a}}}
    \DisplayProof
  \quad
    \AxiomC{$\typing{\vect{t}}{\vect{\gamma}}$}
    \LeftLabel{(constr)}
    \RightLabel{$\simpleconstrdecl{\kappa}{\vect{\gamma}}{\datasort{s}} \in \Sigma$}
        \bottomAlignProof
    \UnaryInfC{\typing{\sconstr{\kappa}{\vect{t}}}{\datasort{s}}}
    \DisplayProof
    \quad
    \AxiomC{$\typing{\vect{t}}{\vect{\gamma}}$}
    \LeftLabel{(abs)}
    \RightLabel{$\atom{a} \in
      \atomset{A}{\atomsort{a}}$}
    \bottomAlignProof
    \UnaryInfC{\typing{\absterm{\atom{a}}{\vect{t}}}{\absort{\atomsort{a}}{\overline{\gamma}}}}      
  \DisplayProof\\[4pt]
}
\noindent These terms are \emph{ground}; i.e., terms without meta-variables~\cite{UrbanC:nomu-jv,FairweatherE:deptnt}. 
As usual in a nominal setting one introduces the notion of freshness and identifies terms 
up to alpha-equivalence. The interesting case of alpha-equivalence is for abstractions, which
is defined by means of permutation of the abstracted atom by a fresh enough one. We omit
those  definitions here, but analogous notions appear in \cref{sect:syntax}.

\subsection{Dependent Sorts}

\subsubsection*{Sort Families}
\label{sec:sort-fam}
We begin by introducing \emph{families of sorts indexed over a sort}. To this end, we extend signatures to include \emph{sort constructors}, i.e., symbols declared as: $\sortconstrdecl{\mathcal{F}}{(\vect{\typing{X}{\gamma}})}$. 

In such a declaration, some sorts $\gamma_j$ might certainly depend on \textit{preceding} parameters $\typing{X_i}{\gamma_i}$ with $i < j$; this is typical of telescopic structures (e.g., contexts in dependently typed lambda calculi). 
Each $X$ in the telescope $\vect{\typing{X}{\gamma}}$ is called a \emph{parameter}; we will use  capital letters $P,Q,X,Y,\ldots$ to denote parameters, and when a parameter is not used in the successive sorts, we shall simply write an underscore, $\typing{\_}{\gamma_i}$.

These
declarations will as a whole take the place of the category of \emph{data} sorts
in the simple setting laid out above.
To declare non-dependent sorts we use the empty list for
$(\vect{\typing{X}{\gamma}})$ and simply write  $\typing{\mathcal{F}}{\data}$.
The generality now introduced allows us to have data sorts obtained by instantiating 
data sort constructors with terms: a sort constructor becomes a sort only when
fully applied to arguments of the appropriate sorts as declared in the signature.

Of course, we will still have term constructors with telescopic parameter structures and result sorts, i.e. declarations of the form \depconstrdecl{f}{\vect{\typing{X}{\gamma}}}{\delta},
with $\delta$ a data sort depending on the parameters $X_i$.

We shall also allow abstraction sorts; in our case, however, the sort of the body of 
an abstraction term might depend on the abstracted atom. Thus, abstraction sorts carry
the name of the atom being abstracted. 
We let atoms carry sorts of \emph{data}, indicating, in a manner to be explained later, the sort of expressions in which they can be instantiated or \emph{concretised}.

It is worth remarking that the system is intended to generate a language of \emph{ground} terms; i.e., terms with parameters are only used to build declarations, and the targeted judgements are all ground. 

\subsubsection*{Motivating Examples}
\label{sec:motee}
We begin with FOL as drawn from a natural deduction
style presentation.
\noindent We start by introducing data sorts for \emph{terms} and \emph{formulas}, as follows:
\begin{alignat*}{2}
    \Trm~: \dta & \qquad \text{and} \qquad &
  \Frm~: \dta
\end{alignat*}
We will for the moment let \Trm\ further unspecified and concentrate on formulas. Each formula has its potential derivations, which we will represent by introducing a family of sorts:

\[
\prf: \Frm \to \data
\]

\noindent A first simple formula is the one representing \emph{falsity},
usually called \emph{bottom} with its elimination rule:
\[ \typing{\perp}{\form} \quad \text{ and } \ \ \depconstrdecl{\perp_{\cnst{e}}}{\typing{\_}{\prf(\perp)},\,\typing{P}{\form}}{\prf(P)}\]
\noindent There are no direct (canonical) derivations of $\perp$, so
we do not introduce any term constructor with target sort $\prf(\perp)$.
The meaning of the declaration for $\perp_{\cnst{e}}$ is as follows: one can 
``instantiate'' the parameters corresponding to the derivation of falsity and 
the formula $P$ to get a proof term for the formula instantiating $P$. 
A vector of arguments \emph{fits} the telescope of a constructor if 
each term of the list has the appropriate sort. 
We now go on to consider \emph{implication}:
\[\depconstrdecl{\impl}{\typing{\_}{\form}, \typing{\_}{\form}}{\form}\]

\noindent Let us consider now its introduction rule; its main parameter must be a
derivation of the consequent (say $Q$) in which an assumption of a derivation of the 
antecedent (say $P$) has been discharged, i.e. made local.
This phenomenon of discharge/locality is naturally represented in our setting by abstraction. More specifically, the sort of the derivation in question will be
an abstraction sort \dabsort{h}{\prf(P)}{\prf(Q)}.
However, a further issue must be considered, namely that the consequent $Q$
can in principle be any term of sort \form. It could contain (free) atoms
whatsoever, in particular $h$. And this is a possibility we wish to exclude,
because it would result in undesired name capture when forming the abstraction
sort above. Hence the declarations in our system have the sort for the
constructor and also a set of freshness conditions, constituting what we call a \emph{freshness context}; this is how we indicate that 
the name (atom) $h$ is to be chosen fresh in the term instantiating $Q$.  
\ We then write:
\[
\fconstrdecl{\impl_{\cnst{i}}}{(\typing{P}{\form}, \typing{Q}{\form},
    \typing{\_}{\dabsort{h}{\prf(P)}{\prf(Q}))}}{\prf(\impl(P,
    Q))}{\fresh{h}{Q}}
\]
Notice that in the declaration of $\impl$ we omitted the freshness context
because it was empty; we will continue with this practice in the rest of the paper.

Let us now show how to actually prove $\varphi \impl \varphi$, for any formula $\varphi$. Indeed, $\prf(\impl(\varphi,\varphi))$ should be realisable for any term $\varphi$ of sort \form, of course without further assumptions. In an ordinary textbook presentation of FOL, a \emph{schema} of derivations would be provided, actually depicting as many concrete derivations as actual formulas there may be. 
We proceed in the same way here, i.e. by offering a \emph{schema} of judgements provable in our system, that correspond to the desired derivations. 
\emph{Internalising} such schemas requires passing from a language of ground terms to one containing meta-variables, a possibility yet to be developed. So let by now suppose we have a (ground) term $\varphi$ of sort \Frm.
Thus,  for some appropriate term ?0 we shall actually derive $\tyjdgmnts{\,?0}{\prf(\impl(\varphi,\varphi))}$.
A possible realisation of ?0 is via the term: $ \impl_{\cnst{i}}\, (\varphi, \varphi, \dabsterm{a}{\prf(\varphi)}{a})$.
Obviously, given any $\varphi$, it should be possible to choose $a$ so as to accomplish the freshness condition indicated above. This is to be checked by the system as explained later.

The elimination rule for implication is now straightforward:
\[\fconstrdecl{\impl_{\cnst{e}}}{(\typing{P}{\form}, \typing{Q}{\form},
    \typing{\_}{\prf(\impl(P, Q))}, \typing{\_}{\prf(P))}}{\prf(Q)}{}\]

\noindent We skip the rest of the usual propositional connectives and go straight to the universal quantifier:
\[
  \forall : (\_ : \absort{\_:\Trm}\Frm) \rightarrow \Frm
\]
This shows another methodological point regarding the encoding of object languages in this system: functions, as e.g. predicates, are uniformly represented as abstractions. We shall have an operation of \emph{atom substitution} and attach to each atom the sort of terms that may be substituted for it. The operation of atom substitution shall be actually subsumed into that of \emph{generalised concretion} (cf.~\cite{pitts:nomlfo-jv}), to be introduced later. An atom shall not in any case have an abstraction sort as its sort, thus making the notion of computation associated to generalised concretion very simple and of a first-order character.

Turning back to the universal quantifier, its introduction  rule requires a proof of a generic instance of the
predicate. Here comes the first use of the 
\emph{concretion} operator, in this case in its original, simple form 
that uses a specific fresh  name instead of the (originally
mute, unknown) abstracted atom:
\[
\forall_i  : (P : \absort{\_{:}\Trm}{\Frm},\_ :\absort{x{:}\Trm}{\prf(P[x])}) \to \prf(\forall (P)) \ ; \ \fresh{x}{P}
\]

\noindent Finally, we consider the elimination of $\mathsf{\forall}$, where $\concr{P}{T}$ below is  an instance of generalised concretion. It corresponds to a
simple concretion on a fresh atom followed by a substitution of this atom by $T$:
\[\\
\forall_e  : (P : \absort{\_{:}\Trm}{\Frm}, T : \Trm, \_ : \prf(\forall(P))) \to \prf(P[T]). 
\]

\section{Syntax}
\label{sect:syntax}
Building on the above, we propose the following syntax as a basis for a dependently sorted system.

\subsection{Grammar}
Consider a countably infinite set of name sorts, each one inhabited by a countably infinite set of names (\emph{atoms}). Let $a, \ b, \ c$, range over atoms.
Let also there be countably infinite sets of \emph{parameters} $X, \ Y, \ Z, \ldots \in \mathbb{X}$;  \emph{term constructors}, $f, \ g \ \ldots \in \mathbb{F}$; and \emph{sort constructors} $\mathcal{F}, \ \mathcal{G}, \ \ldots \in \mathbb{C}$.
Following Gabbay's \emph{permutative convention}~\cite{gabbay:nomtnl}:
$a$, $b$ range over \emph{distinct atoms}. The notation $\vect{t}$ refers to a 
vector of terms $t_0,\ldots,t_n$ with $n\geqslant 0$; given a term $t'$ and $\vect{t} = t_0,\ldots,t_n$, we use $\vect{t},t'$ to represent the vector $t_0,\ldots,t_n,t'$.

Sorts $\gamma$ and terms $t$ are generated by the grammar below. We will use $M$ to stand for either. 
\begin{flalign*}    
  \gamma ::= & \, \sconstr{\mathcal{F}}{\vect{t}} & \textit{}{data~sorts}\\
             & | \, \absort{a : \sconstr{\mathcal{F}}{\vect{t}}}{\gamma}& \textit{abstraction~sorts} \\[2pt]
  t ::= & \, \atom{a}       & atom       \\ 
        & | \, \varTrmD{\vartheta}{\pi}{X} \vect{\concr{}{t
        }} & \textit{\hspace{-4em}parameter with term concretions} \\ 
        & | \, \sconstr{f}{\vect{t}} & \textit{application} \\   
        & | \, \absterm{a : \sconstr{\mathcal{F}}{\vect{t}}}{t}  & \textit{abstraction}
\end{flalign*}

Sorts are built using sort constructors or abstractions and can depend
on terms, which can be atoms, parameters, application of a term constructor to a
tuple of terms, or the abstraction of an atom on a term.
We say an \emph{expression} (sort or term) is \emph{ground} iff it contains no parameters. When a parameter $X$ has no concretions, we omit the square brackets. As stated earlier, parameters are intended for declarations, as shown in the previous section, while sorting judgments (i.e., the language generated by the system to be given in the next section) involve only ground expressions.

\subsection{Operations and Relations}
We define the action of permutations on sorts and terms. Here, $\vect{\prmAct{\pi}{t}}$ denotes 
the vector $\prmAct{\pi}{t_0},\ldots,\prmAct{\pi}{t_n}$.

\begin{definition}[Permutation Action]
  \label{dependent:permutation_action}
  A \emph{permutation} $\pi$ is a bijection on the set of atoms, $\mathbb{A}$, with finite domain.
We represent permutations as lists of swappings  $(a\ b)$. The identity permutation is written $\id$.
  \begin{alignat*}{3}
\prmAct{\pi}{a} & \triangleq \pi(a)
&
\prmAct{\pi}{\absterm{a:\sconstr{\mathcal{F}}{\vect{s}}}{M}} & \triangleq \absterm{\pi(a):\sconstr{\mathcal{F}}{\vect{\pi\act s}}}{(\prmAct{\pi}{M})} 
\\
\prmAct{\pi}{(\varTrmD{\vartheta}{\pi'}{X} \vect{\concr{}{t
    }})} & \triangleq \varTrmD{\prmAct{\pi}{\vartheta}}{(\prmCat{\pi}{\pi'})}{X}\vect{\concr{}{\prmAct{\pi}{t}
}} \ \ &
\prmAct{\pi}{\sconstr{f}{\vect{t}} } & \triangleq \sconstr{f}{\vect{\prmAct{\pi}{t}}}
&
\prmAct{\pi}{\sconstr{\mathcal{F}}{\vect{t}}} & \triangleq \sconstr{\mathcal{F}}{\vect{\pi\act t}}
\end{alignat*}
\end{definition}

\noindent
To define alpha-equivalence, we first introduce the freshness relation. Call $a \frsh M$ a \deffont{freshness constraint}.

\begin{definition}[Freshness Relation]
\label{dependent:freshness} A \emph{freshness judgement} has the form  $\cent a \frsh M$.
 To derive freshness judgements we use the following rules. 
 A premise $ \cent a \frsh  \vect{t}$ is to be expanded as $ \cent a \frsh t_0$, $\ldots$, $\cent a\frsh t_n$.

\begin{center}
\bottomAlignProof
\AxiomC{}
\LeftLabel{\frshRuleName{atm}}
\UnaryInfC{$ \cent a \frsh b$}
\DisplayProof
\qquad
\AxiomC{$\cent a \frsh \vect{t}$}
\LeftLabel{\frshRuleName{cns}}
\bottomAlignProof
\UnaryInfC{$ \cent a \frsh \sconstr{\mathcal{F}}{\vect{t}}$}
\DisplayProof
\qquad
\AxiomC{$ \cent a \frsh \vect{t}$}
\LeftLabel{\frshRuleName{app}}
\bottomAlignProof
\UnaryInfC{$ \cent a \frsh \appTrm{f}{\vect{t}}$}
\DisplayProof
\\[2ex]
\AxiomC{$ \cent a \frsh \sconstr{\mathcal{F}}{\vect{t}}$}
\bottomAlignProof
\LeftLabel{\frshRuleName{ab_{aa}}}
\UnaryInfC{$ \cent a \frsh \absTrmD{a}{\sconstr{\mathcal{F}}{\vect{t}}}{M}$}
\DisplayProof
\quad
\AxiomC{$ \cent a \frsh M$}
\AxiomC{$\cent a \frsh \sconstr{\mathcal{F}}{\vect{t}}$}
\LeftLabel{\frshRuleName{ab_{ab}}}
\bottomAlignProof
\BinaryInfC{$\cent a \frsh \absTrmD{b}{\sconstr{\mathcal{F}}{\vect{t}}}{M}$}
\DisplayProof
\quad
\bottomAlignProof
\AxiomC{$ \cent a \frsh \vect{t}$}
\LeftLabel{\frshRuleName{var}}
\UnaryInfC{$ \cent a \frsh \varTrmD{\vartheta}{\pi}{X} \vect{\concr{}{t}}$}
\DisplayProof
\end{center}
\end{definition}
The main difference with respect to the freshness relation for standard nominal terms is the introduction of new rules $\frshRuleName{ab_{aa}}$,
$\frshRuleName{ab_{ab}}$, $\frshRuleName{cns}$, and $\frshRuleName{var}$, the rule for concretion, which checks freshness only in terms and not in the parameter. 
As will be commented again later, parameters stand for arbitrary \emph{closed} ground terms of the target language.

\begin{definition}[Alpha-equivalence Relation]
\label{dependent:alpha_equivalence}
An  \emph{$\alpha$-equivalence judgement} has the form $\cent M \aleq N$, where $M$ and $N$ are ground. We introduce now the rules defining this relation.
A premise $ \cent \vect{s} \aleq \vect{t}$, where \vect{s} and \vect{t} must always be of the same size, is to be expanded in an element-wise manner into premises $ \cent s_i \aleq t_i$.

\begin{center}
\bottomAlignProof
\AxiomC{}
\LeftLabel{\alphaRuleName{atm}}
\UnaryInfC{$\cent a \aleq a$}
\DisplayProof
\quad
\bottomAlignProof
\AxiomC{$ \cent \vect{s} \aleq \vect{t}$}
\LeftLabel{\alphaRuleName{cns}}
\UnaryInfC{$\cent \sconstr{\mathcal{F}}{\vect{s}} \aleq \sconstr{\mathcal{F}}{\vect{t}}$}
\DisplayProof
\quad
\bottomAlignProof
\AxiomC{$ \cent \vect{s} \aleq \vect{t}$}
\LeftLabel{\alphaRuleName{app}}
\UnaryInfC{$ \cent \appTrm{f}{\vect{s}} \aleq \appTrm{f}{\vect{t}}$}
\DisplayProof
\\[2ex]
\AxiomC{$ \cent \sconstr{\mathcal{F}}{\vect{t}} \aleq \sconstr{\mathcal{F}}{\vect{u}}$}%
\AxiomC{$ \cent M \aleq M'$}
\LeftLabel{\alphaRuleName{ab_{aa}}}
\bottomAlignProof
\BinaryInfC{$ \cent \absTypD{a}{\sconstr{\mathcal{F}}{\vect{t}}}{M} \aleq \absTypD{a}{\sconstr{\mathcal{F}}{\vect{u}}}{M'}$}
\DisplayProof
\\[2ex]
\AxiomC{$\cent \sconstr{\mathcal{F}}{\vect{t}} \aleq \sconstr{\mathcal{F}}{\vect{u}}$}
\AxiomC{$\cent M \aleq \prmAct{\trans{a}{b}}{M'}$}
\AxiomC{$\cent a \frsh M'$}
\LeftLabel{\alphaRuleName{ab_{ab}}}
\bottomAlignProof
\TrinaryInfC{$\cent \absTypD{a}{\sconstr{\mathcal{F}}{\vect{t}}}{M} \aleq \absTypD{b}{\sconstr{\mathcal{F}}{\vect{u}}}{M'}$}
\DisplayProof
\end{center}
\end{definition}
This definition of alpha-equivalence 
generalises  the standard one for nominal terms.
For simplicity, we omit a rule for parameters, which is not essential but would facilitate the writing of declarations.
\begin{lemma}[Equivariance]
\label{lem:equivariance-alpha-fresh}
\hfill
If \ $\cent a\,\#\,M$ then $\cent \pi\act a\ \#\ \pi\act M$.  Similarly if $\cent M\aleq N$ then $\cent \pi\act M\aleq\pi\act N$. 
\end{lemma}

\begin{proof}
Straightforward induction.
\end{proof}

\noindent Freshness is stable by $\alpha$-equivalence:
\begin{lemma}
If $\cent a \frsh M$ and $\cent M \aleq N$ then $\cent a \frsh N$.
\end{lemma}

\begin{proof}
By induction on the freshness relation. Use equivariance.
\end{proof}
\begin{lemma}
$\aleq$ is a congruence.
\end{lemma}
\begin{proof}
    Induction on the definition of $\aleq$.
\end{proof}

\newcommand{\asubs}[2]{\ensuremath{[#1 \mapsto #2]}}

\begin{definition}[Atom Substitution]
\label{dependent:atom_substitution}
We write \asubs{a}{t} for the operation that substitutes the atom $a$ by the term $t$. This is defined on expressions as follows:
\begin{alignat*}{2}
\aSubAct{a}{\asubs{a}{t}} & \triangleq  t & \qquad
\aSubAct{b}{\asubs{a}{t}} & \triangleq  b \\
\aSubAct{(\appTrm{f}{\vect{s}})}{\asubs{a}{t}} & \triangleq  \appTrm{f}{(\vect{\aSubAct{s}{\asubs{a}{t}}})} &
\aSubAct{(\appTrm{\mathcal{F}}{\vect{s}})}{\asubs{a}{t}} & \triangleq  \appTrm{\mathcal{F}}{(\vect{\aSubAct{s}{\asubs{a}{t}}})} \\
\aSubAct{(\varTrmD{\asubs{a}{t}}{\pi}{X} \vect{\concr{}{t}})}{\asubs{a}{t'}} & \triangleq \varTrmD{(\joinTsubs{\vartheta'}{\vartheta})}{\pi}{X} \vect{\concr{}{\aSubAct{t}{\asubs{a}{t'}}}} &  & \\
\aSubAct{(\absTrmD{a}{\sconstr{\mathcal{F}}{\vect{t}}}{M})}{\asubs{a}{t}} & \triangleq \absTrmD{a}{\sconstr{\mathcal{F}}{\vect{\aSubAct{t}{\asubs{a}{t}}}}}{M}  &   \\
\aSubAct{(\absTrmD{b}{\sconstr{\mathcal{F}}{\vect{t}}}{M})}{\asubs{a}{t}} & \triangleq \absTrmD{c}{\sconstr{\mathcal{F}}{\vect{\aSubAct{t}{\asubs{a}{t}}}}}{\aSubAct{(\prmAct{\trans{b}{c}}{M})}{\asubs{a}{t}}}  &  & (\cent c \frsh M, \cent c \frsh t)  . 
\end{alignat*}
\end{definition}

Some explanations are in order: 
to avoid capturing unabstracted atoms,  when an atom substitution acts upon an abstraction or abstraction sort (last case above),
a suitable alpha-equivalent representative of the latter is first chosen.
Any implementation of this definition as a recursive function must accommodate a suitable mechanism for the generation of names; this is most easily achieved by the threading of global state throughout the function or by the use of a global choice function that returns the next available name.\\

\noindent
Atom substitutions work uniformly on alpha-equivalence classes.
\begin{lemma}
\label{dependent:aleq_sub}
If $\vdash M \aleq N$ and $\vdash t \aleq u$ then $\vdash M\asubs{a}{t} \aleq N\asubs{a}{u}$
\end{lemma}
\begin{proof}
Induction on the derivation of $\vdash M \aleq N$.
\end{proof}

A concretion $\concr{w}{t}$ is a partial operation: if $w$ is an abstraction $\absterm{a:\sconstr{\mathcal{F}}{\vect{s}}}{u}$, then its concretion to $t$ evaluates to the body of the abstraction, $u$, where the abstracted atom is substituted by $t$.
If $w$ is the parameter $X$ (possibly with other concretions suspended in it), the concretion remains ``suspended'' (until $X$ is instantiated). Under the sorting system of the next section, concretion of a parameter will be well-sorted  only if the parameter is of an (appropriate) abstraction sort.

\begin{definition}[Concretion]
Concretion is a partial operation:
\begin{alignat*}{2}
\concr{(\absterm{a:\sconstr{\mathcal{F}}{\vect{s}}}{u})}{t} & \triangleq \aSubAct{u}{\aSub{a}{t}{}}  
 & \qquad \qquad
 \concr{(\varTrmD{\vartheta}{\pi}{X} \vect{\concr{}{t}})}{t'} & \triangleq \varTrmD{\vartheta}{\pi}{X} \concr{}{\vect{t},t'}
\end{alignat*}
\end{definition}

\begin{definition}[Parameter Instantiation]
A \emph{parameter instantiation} is a finite mapping from parameters to terms, and it acts on expressions as just grafting (i.e., without a control of capture), subject to the condition that each parameter to be replaced is in the domain of the instantiation.
\end{definition}

\section{Sorting judgements}
\label{sec:sort-rules}
We use five forms of judgements: \begin{inparaenum}[1)]
\item Well-formedness of signature $\Sigma$, formally $\sigjdgmnt{\Sigma}$
(\cref{def:sig-rules});
\item Well-formedness of telescopes $\mathscr{T}$ under a valid signature,  \teljdgmntsig{}{\Sigma}{\mathscr{T}}(\cref{def:tel-form-rules});
\item Well-formedness of contexts of atoms (~\cref{def:ctxt-rules}), $\ctxjdgmntsig{\mathscr{T}}{}{\Sigma}{\Gamma}$;
\item Well-formedness of sorts (\cref{def:sort-rules}), $\sortjdgmntsig{\mathscr{T}}{\Gamma}{}{\Sigma}{\gamma}$; and
\item Well-sortedness of terms (\cref{def:term-rules}), $\tyjdgmntsig{\mathscr{T}}{\Gamma}{}{\Sigma}{t}{\gamma}$.
\end{inparaenum}

As indicated above, the sorts are either \emph{data sorts} or \emph{abstraction sorts}. Data sorts are introduced by \emph{sort constructors} $\mathcal{F}$, and these can only introduce data sorts, never an abstraction sort ---the latter being formed exclusively by the binder $\absort{\typing{\_}{\_}}{\_}$. 
Similarly, terms of the data sorts are formed by \emph{(term) constructors} $f$, and terms of abstraction sorts exclusively by the corresponding binder. 
\emph{Signatures} are sequences of \emph{declarations} of sort and term constructors. As already explained, a declaration specifies the sorts of the corresponding parameters and a freshness context.
These parameter declarations are called \emph{telescopes}. 
The word \emph{context} is reserved for \emph{atom contexts}, $\Gamma$,  necessary to  sort abstractions.

As already stated, the intention is that the system is used for generating well-formed \emph{ground} sorts and
terms. The rules given below define  well-formed scripts of declarations (i.e. signatures), which involve not only ground expressions but also expressions with parameters.

First, notice the use of \emph{freshness contexts} ($\Delta$) in declarations. They involve conditions of the form $a \frsh X$, where the atom $a$ is to appear bound in the declaration and $X$ is any parameter of the declaration.

This defines the side condition on well-formedness of the contexts $\Delta$.
The rules check the validity of the freshness conditions whenever a declaration is put into use, i.e. in rules (data) and (constr). There the constructor employed must be declared in the signature with a telescope $\mathscr{T'}$ and freshness context $\Delta$, as stated in the side condition. Then a fresh version of $\mathscr{T'}$, as well as of $\Delta$, are created by employing new atoms so as to avoid possible collisions with unabstracted atoms in the expression being checked. We call this new telescope $\mathscr{T}_{\#}'$, and the new context $\Delta_{\#}$. Then it is checked that the tuple $\vect{t}$ of arguments \emph{fits} the telescope $\mathscr{T}_{\#}'$ and at the same time the conditions in $\Delta_{\#}$ are satisfied, with the mentioned parameters instantiated accordingly by the tuple $\vect{t}$ ---which we write $(\Delta_{\#})_{\vect{t}}$. 
That a tuple of terms \emph{fits} a telescope has the (obvious) meaning that:
\begin{inparaenum}[a)]
\item The telescopes and the context are well-formed.
\item They are of the same length.
\item Each term has the sort attached to its corresponding parameter, instantiated on the preceding terms in the tuple.
\end{inparaenum}

An equally valid alternative is that the freshness conditions are rather \emph{imposed} by the system, i.e. a freshness declaration is to be interpreted as an \emph{assumption} on part of the user about the employment of names in the (ground) expressions to be generated. The conditions can be imposed by the system by generating in each case a sample chosen among all the alpha-equivalent expressions satisfying the sorting rules that also respects the freshness conditions. For this to work, it is essential that the system is closed under alpha-equivalence ---which  will be shown presently--- and that the freshness conditions are only on bound atoms ---which is already imposed in the well-formation of declarations.

In the rule (constr) we use the notation $\vect{u}_{\vect{t}}$, which stands for the instantiation of the parameters of the tuple of terms $\vect{u}$ with the tuple $\vect{t}$.

Finally, let us remark that, as stated in rule (fun-sig) and (cons-tel), valid telescopes and target sorts of term constructors cannot depend on (unabstracted) atoms. Also note that in the rules we omit premises that can be deduced from some explicitly mentioned premise.

\newcommand{\optional}[1]{}
\def\ScoreOverhang{1pt}
\def\defaultHypSeparation{\hskip 0.08in}
\begin{figure}[h!]
\centering
\begin{subfigure}{\linewidth}
  \centering
\AxiomC{\ }
\LeftLabel{\ruleNameLbl{empty-sig}}
\UnaryInfC{$\sigjdgmnt{\SigEmpty}$}
\DisplayProof
\qquad
\AxiomC{\optional{$\sigjdgmnt{\Sigma}$}}
  \AxiomC{$\teljdgmntsig{}{\Sigma}{\mathscr{T}}$}
  \LeftLabel{\ruleNameLbl{sort-sig}}
  \RightLabel{$\Bigg\{$
    \parbox{1.2in}{$\begin{aligned}
      &\mathcal{F}\not\in\dom{\Sigma}\\
      &\Delta \text{ well-formed} 
    \end{aligned}$}
  }
  \BinaryInfC{$\sigjdgmnt{\SigConsSort{\Sigma}{\mathcal{F}}{\mathscr{T}}{\Delta}}$}
\DisplayProof

\vspace*{8pt}
  \AxiomC{\optional{$\sigjdgmnt{\Sigma}$}}
  \AxiomC{\optional{$\teljdgmntsig{}{\Sigma}{\mathscr{T}}$}}
  \AxiomC{$\clsortjdgmntsig{\mathscr{T}}{}{\Sigma}{\sconstr{\mathcal{F}}{\vect{t}}}$}
  \LeftLabel{\ruleNameLbl{fun-sig}}
  \RightLabel{$\Bigg\{$
      $\begin{array}{l}
        f\not\in\dom{\Sigma}\\
        \Delta \  \text{ well-formed}
      \end{array}$%
  }
  \TrinaryInfC{$\sigjdgmnt{\SigConsFun{\Sigma}{f}{\mathscr{T}}{\sconstr{\mathcal{F}}{\vect{t}}}{\Delta}}$}
  \DisplayProof
  
  \caption{\emph{Rules for signatures.}}
  \label{def:sig-rules}
\end{subfigure}

\vspace*{8pt}
 
\begin{subfigure}{\linewidth}
\centering
\AxiomC{$\sigjdgmnt{\Sigma}$}
\LeftLabel{\ruleNameLbl{empty-tel}}
\UnaryInfC{$\teljdgmntsig{}{\Sigma}{\emptyTel}$}
\DisplayProof
\qquad
\AxiomC{$\clsortjdgmntsig{\mathscr{T}}{}{\Sigma}{\gamma}$}
\LeftLabel{\ruleNameLbl{cons-tel}}
\RightLabel{$X \not\in\dom{\mathscr{T}}$}
\UnaryInfC{$\teljdgmntsig{}{\Sigma}%
{\consTel{\mathscr{T}}{X}{\gamma}}$}
\DisplayProof
  \caption{\emph{Rules for telescope formation.}}
  \label{def:tel-form-rules}
\end{subfigure}

\vspace*{8pt}

\begin{subfigure}{\linewidth}
  \centering
\AxiomC{$\teljdgmntsig{}{\Sigma}{\mathscr{T}}$}
\LeftLabel{\ruleNameLbl{emp-ctx}}
\UnaryInfC{$\ctxjdgmntsig{\mathscr{T}}{}{\Sigma}{\emptyCtx}$}
\DisplayProof
\qquad
\AxiomC{$\sortjdgmntsig{\mathscr{T}}{\Gamma}{}{\Sigma}{\sconstr{\mathcal{F}}{\vect{t}}}$}
  \LeftLabel{\ruleNameLbl{cons-ctx}}
  \RightLabel{$a \not\in\dom{\Gamma}$}
  \UnaryInfC{$\ctxjdgmntsig{\mathscr{T}}{}{\Sigma}{\consCtx{\Gamma}{a}{\sconstr{\mathcal{F}}{\vect{t}}}}$}
  \DisplayProof
\caption{\emph{Rules for well-formed contexts.}}
  \label{def:ctxt-rules}
\end{subfigure}

\vspace*{8pt}

\begin{subfigure}{\linewidth}
\centering
\AxiomC{\optional{$\teljdgmntsig{}{\Sigma}{\mathscr{T}}$}}
\AxiomC{$\telfitsjdgmntsig{\mathscr{T}}{\Gamma}{\Sigma}{\vect{t}}{\mathscr{T}_{\#}'[(\Delta_{\#})_{\vect{t}}]}$}
\LeftLabel{\ruleNameLbl{data}}
\RightLabel{
\begin{minipage}[c]{1.4in}$
  \left\{\begin{array}{l}
     \mathcal{F}\in\dom{\Sigma}\\
     \lookupConstSig{\mathcal{F}}{\Sigma}{\mathscr{T'}\to\data ; \Delta}
     \end{array}\right.$
\end{minipage}
}
\BinaryInfC{$\sortjdgmntsig{\mathscr{T}}{\Gamma}{}{\Sigma}{\sconstr{\mathcal{F}}{\vect{t}}}$}
\DisplayProof
\vspace{8pt}

\AxiomC{$\sortjdgmntsig{\mathscr{T}}{\Gamma}{}{\Sigma}{\sconstr{\mathcal{F}}{\vect{t}}} \ \  \ $}
\AxiomC{$\sortjdgmntsig{\mathscr{T}}{(\Gamma, b :
      \sconstr{\mathcal{F}}{\vect{t}})}{}{\Sigma}{\prmAct{\trans{a}{b}}{\gamma}}$}
 \insertBetweenHyps{\hspace{0pt}}
  \LeftLabel{\ruleNameLbl{abs-*}}
  \RightLabel{\begin{minipage}[c]{1in}
      $\left\{\begin{array}{l}
        b \not\in\dom{\Gamma}\\
        b\frsh \gamma
      \end{array}\right.$
      \end{minipage}
  }
  \BinaryInfC{$\sortjdgmntsig{\mathscr{T}}{\Gamma}{}{\Sigma}{\absort{a : \sconstr{\mathcal{F}}{\vect{t}}}{\gamma}}$}
  \DisplayProof
        
  \caption{\emph{Rules for well-formed sorts.}}\label{def:sort-rules}
\end{subfigure}

\vspace{8pt}

\begin{subfigure}{\linewidth}
\centering
\AxiomC{$\ctxjdgmntsig{\mathscr{T}}{}{\Sigma}{\Gamma}$}
    \LeftLabel{\ruleNameLbl{atm}}
    \RightLabel{$a \in\dom{\Gamma}$}
    \UnaryInfC{$\tyjdgmntsig{\mathscr{T}}{\Gamma}{}{\Sigma}{\atom{a}}{\Gamma(a})$}
    \DisplayProof 
\qquad
\AxiomC{$\ctxjdgmntsig{\mathscr{T}}{}{\Sigma}{\Gamma}$}
    \LeftLabel{\ruleNameLbl{var1}}
    \RightLabel{$X\in\dom{\mathscr{T}}$}
    \UnaryInfC{$\tyjdgmntsig{\mathscr{T}}{\Gamma}{}{\Sigma}{X}{\mathscr{T}(X)}$}
    \DisplayProof 

\vspace{8pt}

\AxiomC{$\tyjdgmntsig{\mathscr{T}}{\Gamma}{}{\Sigma}{\concrVect{X}{t}}{\absort{a :
          \sconstr{\mathcal{F}}{\vect{s}}}{\gamma}}$ \ \ \ }
\AxiomC{$\tyjdgmntsig{\mathscr{T}}{\Gamma}{}{\Sigma}{t'}{\sconstr{\mathcal{F}}{\vect{s}}}$}
    \LeftLabel{\ruleNameLbl{var2}}
    \BinaryInfC{$\tyjdgmntsig{\mathscr{T}}{\Gamma}{}{\Sigma}{\concr{X}{\vect{t},t'}}{\gamma[a\mapsto t']}$}
\DisplayProof

\vspace{10pt}

\AxiomC{$\telfitsjdgmntsig{\mathscr{T}}{\Gamma}{\Sigma}{\vect{t}}{\mathscr{T}_{\#}'[(\Delta_{\#})_{\vect{t}}]}$}
\LeftLabel{\ruleNameLbl{constr}}
\RightLabel{
  \parbox{1.4in}{
    \[\left\{\begin{array}{l}
      f \in \dom{\Sigma}\\
      \lookupConstSig{f}{\Sigma}{\mathscr{T'}\to \sconstr{\mathcal{F}}{\vect{u}} ; \Delta} \\
    \end{array}\right.\]
  }}
\UnaryInfC{$\tyjdgmntsig{\mathscr{T}}{\Gamma}{}{\Sigma}{\sconstr{f}{\vect{t}}}{\sconstr{\mathcal{F}}{(\vect{u}_{\vect{t}})}}$}
\DisplayProof

\vspace{8pt}

\AxiomC{$\sortjdgmntsig{\mathscr{T}}{\Gamma}{}{\Sigma}{\sconstr{\mathcal{F}}{\vect{t}}}$ \ \ \ }
\AxiomC{$\tyjdgmntsig{\mathscr{T}}{(\Gamma, b :
        \sconstr{\mathcal{F}}{\vect{t}})}{}{\Sigma}{\prmAct{\trans{a}{b}}{t}}{\prmAct{\trans{a}{b}}{\gamma}}$}
\LeftLabel{\ruleNameLbl{abs}}
\RightLabel{
      \begin{minipage}[c]{1in}$%
      \left\{\begin{array}{l}
          b \not\in\dom{\Gamma}\\
          b\frsh \{t, \gamma\}
        \end{array}\right.$
      \end{minipage}}
\BinaryInfC{$\tyjdgmntsig{\mathscr{T}}{\Gamma}{}{\Sigma}{\absterm{a :\sconstr{\mathcal{F}}{\vect{t}}}{t}}{\absort{a : \sconstr{\mathcal{F}}{\vect{t}}}{\gamma}}$}
\DisplayProof
    
\vspace{8pt}

\AxiomC{$\tyjdgmntsig{\emptyTel}{\Gamma}{}{\Sigma}{t}{\gamma}$}
\AxiomC{\optional{$\sortjdgmntsig{\emptyTel}{\Gamma}{}{\Sigma}{\gamma'}$}}    \LeftLabel{\ruleNameLbl{conv}}
    \RightLabel{$\gamma \aleq \gamma'$}
\BinaryInfC{$\tyjdgmntsig{\emptyTel}{\Gamma}{}{\Sigma}{t}{\gamma'}$}
\DisplayProof
   
\caption{\emph{Rules for well-sorted terms.}}
\label{def:term-rules}
\end{subfigure}
\caption{Sorting System}
\end{figure}

\subsection*{Properties of the sorting system}

\begin{lemma}[Equivariance of sorting]
\label{sorting-equivariant}
    Let $\mathcal{J}'$ be a permutative variant of the judgement $\mathcal{J}$. If one is derivable, then the other is also derivable.
\end{lemma}
\begin{proof}
    Direct check on the system by simultaneous induction over all forms of judgments.
\end{proof}
\begin{lemma}[Closure under alpha conversion]
\label{alpha-closure} \hfill
\begin{enumerate}
    \item If $\sortjdgmntsig{\emptyTel}{\Gamma}{}{\Sigma}{\gamma}$
    then for every $\gamma'$ such that $\vdash \gamma' \aleq \gamma$, also 
$\sortjdgmntsig{\emptyTel}{\Gamma}{}{\Sigma}{\gamma'}.$
\item If $\tyjdgmntsig{\emptyTel}{\Gamma}{\Delta}{\Sigma}{t}{\gamma}$ 
then for every $t'$ such that $\vdash t' \aleq t$, also 
$\tyjdgmntsig{\emptyTel}{\Gamma}{\Delta}{\Sigma}{t'}{\gamma}$.
\end{enumerate}
\end{lemma}
\begin{proof}
By induction on the sorting system, essentially using rule (conv) and equivariance.
\end{proof}
\noindent Since sort inference and alpha-equivalence are decidable, sort-checking is decidable.
\begin{lemma}[Decidability]
\label{decidability}
Given a signature $\Sigma$, telescope $\mathscr{T}$ and context $\Gamma$ valid under $\Sigma$, it is decidable whether $M$ is a sort or a term of some sort, for any $M$.
\end{lemma}

\section{Representation of calculi}

\subsection{First Order Logic}
\label{sec:FOL}

Normally one introduces \emph{first-order} languages, each one determined by a choice of function and predicate symbols.  
This would then call for a specification parameterised on such symbol declarations.
While this kind of parameterisation could be incorporated into the framework, we prefer for now to keep matters simple and provide a representation of \emph{first-order arithmetic}. 
Another choice we make is to represent a \emph{classical} version of the logic.

\subsubsection*{Syntax}

\paragraph{Presentation.}
In the following $x$ represents a variable taken from a denumerable set.
\begin{flalign*}
  t ::=~& x \mid 0 \mid S~t \mid t_1 + t_2 \mid t_1 \times t_2 & \text{Terms} \\
  \varphi ::=~& t_1 = t_2 \mid \bot \mid \neg \varphi \mid \varphi_1 \supset \varphi_2  \mid (\forall x)\varphi & \text{Formulæ} 
\end{flalign*}
The notion of free variable for terms and formulae $fv(t)$ and $fv(\varphi)$ are defined as usual.

Let us call \emph{expressions} (denoted by $e$) either terms or formulae of the first-order language being considered. We write $\equiv$ for identity on terms and formulas. This is the congruent-closure of $\alpha$-conversion, which in turn is the diagonal on terms and defined by the following rule on formulae:
\begin{center}
\AxiomC{$\varphi^x_z \equiv \varphi'^y_z$}
      \RightLabel{$z \notin fv(\varphi)$}
      \UnaryInfC{$(\forall x)\varphi \equiv (\forall y)\varphi'$}
      \DisplayProof    
\end{center}

\noindent Here $\varphi^x_z$ is the \emph{swapping} of the occurrences of $x$ and $z$ in $\varphi$.
Notice $\equiv$ is clearly decidable.  Now define substitution variable by a term on terms and formulas, in the usual way, using the clause:
\[
((\forall x)\varphi) [y := t] \equiv (\forall x)(\varphi [y := t]) \qquad x \notin fv(t \cup \{y\})
\]
\noindent Justification of well-definedness of substitutions is routine, assuming the clause:
\[
\varphi[x := t] \equiv \varphi'[x := t] \ \text{ if } \ \varphi \equiv \varphi'\footnote{This is the same as identifying formulae up to $\alpha$-conversion, or working on $\alpha$-classes.}\]  
This level of detail is usually not reached in textbook presentations, and actually neither in e.g.~\cite{harper:fradl}. Other alternatives could have been chosen – this one is direct and simple enough.

\paragraph{Encoding.}
Clearly, we have two \emph{sorts} of expressions:
\begin{alignat*}{2}
    \Trm~: \dta & \qquad \text{and} \qquad &
  \Frm~: \dta
\end{alignat*}
Look now at terms. Each variable $x_i$ will be encoded as a (distinct) atom $a_i$ (that is to carry the sort \Trm).
Besides, we introduce constructors and operations:
\begin{alignat*}{2}
  0 & : \Trm & \qquad
  S & : (\_:\Trm) \rightarrow \Trm \\
  + & : (\_:\Trm, \_:\Trm) \rightarrow \Trm &
  \times & : (\_:\Trm, \_:\Trm) \rightarrow \Trm
\end{alignat*}
\noindent As usual, we overload the symbols on the object- and meta-levels. As to formulae:
\begin{alignat*}{3}
  = & : (\_:\Trm, \_:\Trm) \rightarrow \Frm &\qquad
  \bot & : \Frm &\qquad
  \neg & : (\_:\Frm) \rightarrow \Frm \\
  \supset & : (\_:\Frm, \_:\Frm) \rightarrow \Frm &
  \forall & : (\absort{\_:\Trm}\Frm) \rightarrow \Frm &
\end{alignat*}
\noindent
Notice that binding in the object language is represented using the abstraction construct of the framework.

\paragraph{Adequacy.}

 Let us call $\Sigma$ the signature just introduced.
For any finite set $\mathcal{X}$ of variables $x_1, \ldots, x_n$, let $\hat{\mathcal{X}}$ be a context of our nominal framework (call this NF) containing assumptions $a_i : \mathsf{Term}$ iff $x_i \in \mathcal{X}$.

\noindent
\begin{lemma}\label{LAdeq:FOL}
There is a compositional bijection\footnote{Compositional means that substitution commutes with encoding, as introduced in \cite{harper:fradl}.} between:
\begin{itemize}
    \item Terms $t$ of the FOL-calculus and terms $\hat{t}$ of the NF such that $\tyjdgmntsig{\emptyCtx}{\widehat{fv(t)}}{\emptyCtx}{\Sigma}{\hat{t}}{\mathsf{Term}}{}$
    \item Formulae $\varphi$ of the FOL-calculus and terms $\hat{\varphi}$ of the NF such that $\tyjdgmntsig{\emptyCtx}{\widehat{fv(\varphi)}}{\emptyCtx}{\Sigma}{\hat{\varphi}}{\mathsf{Form}}{}$
\end{itemize}
\end{lemma}

\begin{proof}
First we define encoding \enc\ as follows, by recursion on FOL-terms (we use $\approx$ for term identity in NF, which includes $\aleq$):
\begin{alignat*}{2}
\enc(x_i) & \approx \  a_i     & \qquad \enc(t_1 \ = \  t_2) & \approx \ =(\enc(t_1), \enc(t_2)) \\
\enc(0) & \approx \  0         &    \enc(\bot) & \approx \  \bot \\
\enc(S\,t) & \approx \  S(\enc(t)) & \enc(\neg\,\varphi) & \approx \  \neg(\enc(\varphi)) \\
\enc(t_1 + t_2) & \approx \  +(\enc(t_1), \enc(t_2)) & \enc(\varphi_1 \supset \varphi_2) & \approx \  \supset(\enc(\varphi_1), \enc(\varphi_2)) \\
\enc(t_1 \times t_2) & \approx \  \times(\enc(t_1), \enc(t_2)) & \enc((\forall x_i)\varphi) & \approx \  \forall (\absort{a_i:\Trm}\enc(\varphi))
\end{alignat*}

\noindent
Next, we show that $e_1 \equiv e_2 \iff \emptyCtx \vdash \enc(e_1) \approx \enc(e_2)$
 (notice that \enc\ is ground for every $e$, so there's no need to consider freshness contexts on the right-hand side).

\noindent
Define \dec, the inverse to \enc, as follows:
\begin{alignat*}{2}
\dec(a_i) & \equiv \ x_i & 
    \dec(=(\hat{t}_1, \hat{t}_2)) & \equiv \ \dec(\hat{t}_1) = \dec(\hat{t}_2) \\
\dec(0) & \equiv \ 0 & 
    \dec(\bot) & \equiv \ \bot \\
\dec(S(\hat{t})) & \equiv \ S(\dec(\hat{t})) & 
    \dec(\supset(\hat{\varphi}_1, \hat{\varphi}_2)) & \equiv \ \dec(\hat{\varphi}_1) \supset \dec(\hat{\varphi}_2) \\
\dec(+(\hat{t}_1, \hat{t}_2)) & \equiv \ \dec(\hat{t}_1) + \dec(\hat{t}_2) & 
    \qquad\dec(\forall(\absort{a_i:\Trm}\hat{\varphi})) & \equiv \ (\forall x_i)\ \dec(\hat{\varphi}) \\
\dec(\times(\hat{t}_1, \hat{t}_2)) & \equiv \ \dec(\hat{t}_1) \times \dec(\hat{t}_2) & 
\phantom{2} & \phantom{=} \phantom{b}
\end{alignat*}
\noindent
It follows \enc\ is a bijection with inverse \dec. It is straightforward to prove, by induction on terms and formulae, that the sorting judgments hold using the signature given above.\\

\noindent Compositionality is given by:
$\enc(e[x_i{:=}t]) \approx (\enc(e))[a_i \mapsto \enc(t)]$
\end{proof}

\subsubsection*{Derivations}
\paragraph{Presentation.}
We choose Natural Deduction. Let contexts $\Gamma$ of assumptions be finite sets of formulas. Write $\Gamma, \varphi$ for $\Gamma \cup \{\varphi\}$ with $\varphi \notin \Gamma$, and extend $fv$ to contexts in the obvious way.
\begin{center}
 \centering
 \AxiomC{\ }
   \LeftLabel{(ass)}
   \RightLabel{$\varphi \in \Gamma$}
   \UnaryInfC{$\Gamma \vdash \varphi$}
\DisplayProof
\qquad
\AxiomC{\mbox{  }}
   \LeftLabel{($\rho$)}
   \UnaryInfC{$\Gamma \vdash t = t$}
   \DisplayProof
 \qquad
 \AxiomC{$\Gamma \vdash t_1 = t_2$}
 \AxiomC{$\Gamma \vdash \varphi[x_i := t_1]$}
   \LeftLabel{($\sigma$)}
   \BinaryInfC{$\Gamma \vdash \varphi[x_i := t_2]$}
   \DisplayProof

\hfill

  \AxiomC{}
    \LeftLabel{($+_0$)}
    \UnaryInfC{$\Gamma \vdash 0 + t = t$}
    \DisplayProof
\qquad
  \AxiomC{}
    \LeftLabel{($+_S$)}
    \UnaryInfC{$\Gamma \vdash s(t_1) + t_2 = s(t_1 + t_2)$}
    \DisplayProof

\hfill

 \AxiomC{}
   \LeftLabel{($\times_0$)}
   \UnaryInfC{$\Gamma \vdash 0 \times t = 0$}
   \DisplayProof
\qquad
 \AxiomC{}
   \LeftLabel{($\times_S$)}
   \UnaryInfC{$\Gamma \vdash s(t_1) \times t_2 = (t_1 \times t_2) + t_2$}
   \DisplayProof

\hfill

 \AxiomC{$\Gamma \vdash \varphi[x := 0]$}
 \AxiomC{$\Gamma,\varphi \vdash \varphi[x := s x]$}
 \LeftLabel{(ind)}
 \RightLabel{$x \notin FV(\Gamma)$}
 \BinaryInfC{$\Gamma \vdash \varphi[x := t]$}
 \DisplayProof

\vspace{6pt}

 \AxiomC{$\Gamma \vdash \bot$}
   \LeftLabel{($\bot_e$)}
   \RightLabel{$\varphi \in \Gamma$}
   \UnaryInfC{$\Gamma \vdash \varphi$}
   \DisplayProof
\qquad
 \AxiomC{$\Gamma, \varphi \vdash \bot$}
   \LeftLabel{($\neg$ i)}
   \UnaryInfC{$\Gamma \vdash \neg \varphi$}
   \DisplayProof
 \qquad
 \AxiomC{$\Gamma \vdash \varphi$}
 \AxiomC{$\Gamma \vdash \neg \varphi$}
   \LeftLabel{($\neg$ e)}
   \BinaryInfC{$\Gamma \vdash \bot$}
   \DisplayProof

\vspace{8pt}

 \AxiomC{$\Gamma, \neg \varphi \vdash \bot$}
   \LeftLabel{(RAA)}
   \UnaryInfC{$\Gamma \vdash \varphi$}
   \DisplayProof
\qquad
 \AxiomC{$\Gamma,\varphi_1 \vdash \varphi_2$}
   \LeftLabel{($\supset$i)}
   \UnaryInfC{$\Gamma \vdash \varphi_1 \supset \varphi_2$}
   \DisplayProof
\qquad
 \AxiomC{$\Gamma \vdash \varphi_1\supset \varphi_2$}
 \AxiomC{$\Gamma \vdash \varphi_1$}
   \LeftLabel{($\supset$e)}
   \BinaryInfC{$\Gamma \vdash \varphi_2$}
   \DisplayProof

\vspace{8pt}
 
\AxiomC{$\Gamma \vdash \varphi$}
  \LeftLabel{($\forall$i)}
  \RightLabel{$x_i \notin FV(\Gamma)$}
  \UnaryInfC{$\Gamma \vdash (\forall x_i)\varphi$}
  \DisplayProof
\quad
 \AxiomC{$\Gamma \vdash (\forall x)\varphi$}
   \LeftLabel{($\forall$e)}
   \UnaryInfC{$\Gamma \vdash \varphi[x := t]$}
   \DisplayProof
\end{center}
\paragraph{Encoding.}
Derivations of judgements $\Gamma \vdash \varphi$ can be seen as derivations of judgements $\varphi$ proceeding from assumptions consisting of single formulae, in such a way that the set of assumptions is included in $\Gamma$. We introduce a family of sorts \ 
$\prf: \Frm \to \data$ \ 
 for classifying these latter derivations.

So each derivation of the formula $\varphi$ will be represented as a term in $\prf(\varphi)$. The rules of the system shall be term-formers of these various sorts, in the following way: \begin{inparaenum}[1)]
    \item {Assumptions}: Each assumption of formula $\varphi$ shall be represented as a distinct atom $h:\prf(\varphi)$. These have to be chosen also distinct from the ones representing free variables of the formulae participating in any derivation.
\item{Premises}: are parameters of the corresponding sorts. 
\item{Discharge}: corresponds to binding, i.e., the discharged assumption(s) are bound to the rule(-occurrence) that discharges it (them).
\item{Additional parameters} like terms, must be appropriately declared as parameters of the corresponding constructor. One special case is variables chosen to be replaced (i.e. in rule $\sigma$ or (ind)). These are encoded as atoms possibly appearing in the encoding of the relevant formula. Clearly, in this case the variable is merely a pointer to a place in the formula where to perform the substitution of the relevant term. Accordingly, we encode this phenomenon using our abstraction operator of the framework. That is to say that the relevant formula and the variable in question are encoded as an abstraction.
\end{inparaenum}

In the following encoding, we collapse declarations of parameters of the same sort (so, instead of
writing $t_1 : \Trm, t_2: \Trm$ we write $t_1, t_2 : \Trm$).
\begingroup
\small
\begin{longtable}{ll}
\multicolumn{2}{l}{$\rho : (t : \Trm) \to \prf(=(t,t))$}\\[4pt]
$+_0 : (t : \Trm) \to \prf(=(+(0,t),t))$ & 
$+_S : (t_1, t_2 : \Trm) \to \prf(=(+(S(t_1),t_2)),S(+(t_1,t_2)))$\\[4pt]
$\times_0 : (t : \Trm) \to \prf(=(\times(0,t),0))$ &
$\times_S : (t_1, t_2 : \Trm) \to \prf(=(\times(S(t_1),t_2)),(+(\times(t_1,t_2),t2)))$\\[4pt]
\multicolumn{2}{l}{$\sigma : (\varphi {:} \absort{\_{:}\Trm}\Frm,t_1,\ t_2:\Trm,\_:\prf(=\!(t_1,t_2)),\_ :\prf(\varphi[t_1])) \to \prf(\varphi[t_2])$} \\[4pt]
\multicolumn{2}{l}{$\begin{aligned}\text{ind} : (\varphi {:}\absort{\_{:}\Trm}{\Frm},P0:&\prf(\varphi[0]),\\PS{:}&\absort{x{:}\Trm}{\absort{\_{:}\prf(\varphi[x])}{\prf(\varphi[S(x)])}},t{:}\Trm)
    \to \prf(\varphi[t]) \ ; \ \fresh{x}{\varphi}
\end{aligned}$}\\[6pt]
$\bot_e : (\varphi : \Frm, \_:\prf(\bot)) \to \prf(\varphi)$ &
$\neg_i : (\varphi {:} \Frm,\_:\absort{\_{:}\prf(\varphi)}{\prf(\bot)}) \to \prf(\neg(\varphi))$\\[4pt]
\multicolumn{2}{l}{$\neg_e  : (\varphi : \Frm, \_ :\prf(\varphi), \_ :\prf(\neg(\varphi))) \to \prf(\bot)$}\\[4pt]
\multicolumn{2}{l}{$\supset_i : (\varphi_1, \varphi_2 : \Frm, \_:\absort{\_{:}\prf(\varphi1)}{\prf(\varphi2)}) \to \prf(\supset(\varphi_1,\varphi_2))$}\\[4pt]
\multicolumn{2}{l}{$\supset_e : (\varphi_1, \varphi_2 : \Frm, \_ : \prf(\supset(\varphi_1,\varphi_2)), \_ : \prf(\varphi_1)) \to \prf(\varphi_2)$}\\[4pt]
\multicolumn{2}{l}{$\forall_i  : (\varphi : \absort{\_{:}\Trm}{\Frm},\_ :\absort{x{:}\Trm}{\prf(\varphi[x])}) \to \prf(\forall (\varphi)) \ ; \ \fresh{x}{\varphi}$}\\[4pt]
\multicolumn{2}{l}{$\forall_e  : (\varphi : \absort{\_{:}\Trm}{\Frm}, t : \Trm, \_ : \prf(\forall(\varphi))) \to \prf(\varphi[t])$}
\end{longtable}
\endgroup
\paragraph{Adequacy.}
To prove adequacy of derivations we need to map derivation trees of judgements $\Gamma \vdash \varphi$ to terms  of sort $\prf(\hat{\varphi})$ in the framework. 
These terms are to depend on contexts containing declarations for atoms of sort \Trm\ in $\hat{\varphi}$ (which correspond to free variables in $\varphi$) and for atoms corresponding to the undischarged assumptions of the derivation. 
These, in turn, must be preceded by declarations of atoms of sort \Trm\ that correspond to free variables of the formulas in $\Gamma$.
Besides, the atoms corresponding to the assumptions can be chosen somewhat arbitrarily, as long as they are different enough to ensure well-formation of the context in question.
Thus, for each derivation of judgement $\Gamma \vdash \varphi$ and set of atoms $H$ of size equal to $\Gamma$, we can form a context $\widehat{fv(\Gamma)}, \Gamma_H$,  where $\widehat{fv(\Gamma)}$ is as before (adequacy of terms) and $\Gamma_H$ associates to each $h \in H$ the sort $\prf(\hat{\varphi})$ for $\varphi$ a different formula in $\Gamma$.
We require $H$ disjoint from the set of atoms declared in $\widehat{fv(\Gamma)}$. 
Now we can formulate:

\begin{lemma}
For each set $H$ of atoms such that $|H| = |\Gamma|$, there exists a bijective correspondence between derivations $\delta$ of judgements $\Gamma \vdash \varphi$ in FOL and terms $\hat{\delta}$ of NF such that: \ \tyjdgmnt{\emptyCtx}{\widehat{fv(\Gamma)}}{\emptyCtx}{\hat{\delta}}{\prf(\hat{\varphi})}.
\end{lemma}

\noindent
Let us make two remarks. First, we are still considering only ground terms of the framework.
Second,  we consider identity of derivations in FOL to be given freely by the rules-as-constructors.

\begin{proof}
By induction on derivations.
We define $\hat{\delta}$ as $\delta$\enc\ with parameter $H$ as follows:
\begingroup
\small
\begin{longtable}{llr}
\toprule
Derivation, $\delta$ & Encoding, $\deH(\delta)$ & Comment\\
\midrule
   \AxiomC{}
    \LeftLabel{(ass)}
    \UnaryInfC{$\Gamma \vdash\varphi$}
    \DisplayProof & 
    $h$\text{ where }
    $(h:\prf(\hat{\varphi}) \in \Gamma_H)$&\\[4pt]
  \AxiomC{}
   \LeftLabel{($\rho$)}
    \UnaryInfC{$\Gamma \vdash t=t$}\DisplayProof & 
    $\rho(\hat{t})$ & \\[4pt]
  \AxiomC{$\overset{\delta_1}{\Gamma \vdash t_1 = t_2}$}
  \AxiomC{$\overset{\delta_2}{\Gamma \vdash \varphi[x_i := t_1]}$}
  \LeftLabel{($\sigma$)}
  \BinaryInfC{$\Gamma \vdash \varphi[x_i := t_2]$}\DisplayProof & 
  \multicolumn{2}{l}{$
      \sigma(\absort{x{:}\Trm}{\hat{\varphi}},\ \hat{t_1},\ \hat{t_2},
      \deH(\delta_1), \deH(\delta_2))$}
  \\[14pt]
  \AxiomC{ }
  \LeftLabel{($+_0$)}
  \UnaryInfC{$\Gamma \vdash 0+t=t$} \DisplayProof
  &
  $+_0(\hat{t})$ & \\[-12pt]
  & \multicolumn{2}{r}{
  similarly for $+_S,\times_0,\times_S$}\\[10pt]
  \AxiomC{$\overset{\delta_0}{\Gamma \vdash \varphi[x:=0]}$}
  \AxiomC{$\overset{\delta_1}{\Gamma,\varphi \vdash \varphi[x:=S(x)]}$}
  \LeftLabel{(ind)}
  \BinaryInfC{$\Gamma \vdash \varphi[x:=t]$}
  \DisplayProof&
  $\begin{aligned}
      \texttt{ind}(&\absort{a_i{:}\Trm}{\hat{\varphi}},\deH(\delta_0),\\
      &\absort{a_i{:}\Trm}{\absort{h{:}\prf(\hat{\varphi})}{}}\delta\enc_{H\cup\{h\}}(\delta_1),t)
  \end{aligned}$ &
  $h\not\in H$\\[14pt]
  \AxiomC{$\overset{\delta}{\Gamma \vdash \bot}$}
  \LeftLabel{($\bot_e$)}
  \UnaryInfC{$\Gamma \vdash \varphi$}
  \DisplayProof&
  $\bot_e(\hat{\varphi},\deH(\delta))$ & \\[-20pt] 
  \multicolumn{3}{r}{similarly for $\neg_e, \supset_e$}\\[4pt]
  \AxiomC{$\overset{\delta}{\Gamma,\varphi \vdash \bot}$}
  \LeftLabel{($\neg_i$)}
  \UnaryInfC{$\Gamma \vdash \neg\varphi$}
  \DisplayProof&
  $\neg_i(\hat{\varphi},\absort{h{:}\prf(\varphi)}{\delta\enc_{H\cup\{h\}}(\delta)})$
  &$h \notin H$\\[-8pt]
  \multicolumn{3}{r}{similarly for $\text{RAA},\supset_i$}\\
  \AxiomC{$\overset{\delta}{\Gamma \vdash \varphi}$}
  \LeftLabel{($\forall_i$)}
  \RightLabel{$x_i\notin fv(\Gamma)$}
  \UnaryInfC{$\Gamma \vdash (\forall x_i)\varphi$} 
  \DisplayProof & \multicolumn{2}{l}{
  $\forall_i(\absort{a_i{:}\Trm}{\hat{\varphi}},\absort{a_i{:}\Trm}{\deH(\delta)})$}
  \\[12pt]
  \AxiomC{$\overset{\delta}{\Gamma \vdash (\forall x_i)\varphi}$}
   \LeftLabel{($\forall_e$)} 
  \UnaryInfC{$\Gamma \vdash \varphi[x_i:=t]$}
  \DisplayProof
  & $\forall_e(\absort{a_i{:}\Trm}{\hat{\varphi}}, \hat{t},\deH(\delta))$&\\
  \bottomrule
\end{longtable}
\endgroup
\noindent
Now we show that the typing restriction holds in the framework. We do it for two interesting cases, namely assumption and induction:
\begin{enumerate}
\item[(ass)]We need to show $\widehat{fv(\Gamma)}, \Gamma_H \vdash h : \prf(\hat{\varphi})$. But $h :   \prf(\hat{\varphi})$ is in $\Gamma_H$ by construction, and the whole context is well-formed by  construction too.
\item[(ind)] First of all, \absort{a_i{:}\Trm}{\hat{\varphi}} is of sort \absort{a_i{:}\Trm}{\Frm}.
    Also, by induction hypothesis, $\deH(\delta_0)$ is of sort $\prf(\widehat{\varphi[x_i := 0]})$ 
    under $\widehat{fv(\Gamma)},\Gamma_H$. 
    We get  $\widehat{\varphi[x:=0]} \approx \hat{\varphi}[a\mapsto 0] 
    \approx (\absterm{a_i{:}\Trm}{\hat{\varphi}})[0]$ by Lemma~\ref{LAdeq:FOL}. Therefore the second argument for the constructor \texttt{ind} is well-typed. Next, let us check that
    $\absort{a_i{:}\Trm}{\absort{h{:}\prf(\hat{\varphi})}{\delta\enc_{H\cup\{h\}}(\delta_1)}}$, has sort
    $\absort{a_i{:}\Trm}{\absort{h{:}\prf(\hat{\varphi})}{\prf(\widehat{\varphi[x_i:=S(x_i)]})}}$.
    Then we know both $\hat{\varphi} \ \approx \ (\absort{a_i{:}\Trm}{\hat{\varphi}})[a_i]$ and
    $\widehat{\varphi[x_i := S(x_i)]} \ \approx \ 
    \hat{\varphi}[a_i\mapsto S(a_i)] \ \approx \ 
    (\absort{a_i{:}\Trm}{\hat{\varphi}})[S(a_i)]$, as desired.
\end{enumerate}
Clearly, the mapping is one-to-one. The converse is shown by defining 
a decoding mapping. The idea is to map a term $\hat{\delta}$ such that
$\tyjdgmnt{\emptyCtx}{\widehat{fv(\Gamma)}}{\emptyCtx}{\hat{\delta}}{\prf(\hat{\varphi})}$
for some $\hat{\varphi}$ such that 
$\tyjdgmnt{\emptyCtx}{\hat{\Gamma}}{\emptyCtx}{\hat{\varphi}}{\Frm}$, 
to a derivation of $\Gamma \vdash \varphi$ for appropriate $\Gamma$.
\end{proof}

\subsection{Lambda Calculi}
\label{sec:lambda-shallow}
\subsubsection*{The Pure Calculus}

Here we give a shallow encoding (see the deep version below).
We have to introduce a sort $\Lambda: \dta$ for terms,
after which we get the atoms, which in the present version will represent the usual variables of the calculus. Then it remains to declare:
\[@  : (\_ : \Lambda,\; \_ : \Lambda) \rightarrow \Lambda \ \text{ (for application) \  
and \ } \lambda  : (\_ : \absort{\_{:}\Lambda}{\Lambda}) \rightarrow \Lambda \ \text{ (for functional abstraction).}\]

\noindent
Note that $\beta$-contraction  of redexes, usually denoted by $\triangleright$, is a binary relation to be encoded as:
\[
\triangleright : (\_ : \Lambda,\; \_ : \Lambda) \rightarrow \dta
\]
\noindent
with one rule: \hspace*{6em}
$\beta : (B : \absort{\_{:}\Lambda}{\Lambda},N:\Lambda) \rightarrow \triangleright(@(\lambda(B), N), B[N])$,

\noindent
using the generalised concretion of the framework.
\paragraph{Adequacy.} To show the correctness of our encoding of the $\lambda$-calculus, we need to prove that we can map $\beta$-reductions $s \rightarrow_\beta t$ in the $\lambda$-calculus to terms of type \ $\triangleright(\enc(s),\enc(t))$ in our framework.

\begin{lemma}
If $s,t$ are $\lambda$-terms such that $s \rightarrow_\beta t$, i.e., $s \equiv (\lambda x.s_1)s_2$ and $t \equiv s_1[x:=s_2]$ then 
\\
$\vdash \beta(\absort{a_i{:}\Lambda}{\enc(s_1)} ,\enc(s_2)): \triangleright(@(\lambda(\absort{a_i{:}\Lambda}{\enc(s_1)}), \enc(s_2)), (\absort{a_i{:}\Lambda}{\enc(s_1)})[\enc(s_2)]  $.
\end{lemma}

\subsubsection*{The Simply Typed Calculi}

We give a system in Church's monomorphic style. We introduce a sort $\LATy : \dta$ for types with constructors:
\begin{align*}
\iota &: \LATy &  \text{— any ground type}\\
\Rightarrow &: (\_ : \LATy,\_ : \LATy) \rightarrow \LATy & \text{— the functional types}
\end{align*}
Typed terms is a sort family $\LTrm : (\_ : \LATy) \rightarrow \dta$ indexed by types, with
constructors:
\begin{align*}
@ &: (\alpha : \LATy, \beta : \LATy, \_ : \LTrm(\Rightarrow(\alpha, \beta)),\_ : \LTrm(\alpha)) \rightarrow \LTrm(\beta)
\end{align*}
\begin{align*}
\lambda &: (\alpha:\LATy, \beta:\LATy,\_:\absort{\_{:}\LTrm(\alpha)}{\LTrm(\beta)}) \rightarrow \LTrm(\Rightarrow(\alpha, \beta))
\end{align*}
Contraction gets typed (because generalized concretion of the framework preserves typing):
\begin{align*}
\triangleright &: (\beta : \LATy, \_ : \LTrm(\beta), \_ : \LTrm(\beta)) \rightarrow \dta\\
\beta &: (\alpha : \LATy, \beta : \LATy, B : \absort{\_{:} \LTrm(\alpha)}{\LTrm(\beta)},\, N : \LTrm(\alpha)) \\
    & \quad \rightarrow \triangleright(\beta,@(\alpha,\beta,\lambda(\alpha,\beta,B),N), B[N])
\end{align*}

There is a difficulty with encoding Curry’s style system --- which has to do with the representation of contexts assigning types to variables. Indeed, since variables in our encoding are atoms of sort $\Lambda$ in the framework context, we would need to somehow zip the context to a list of $\LATy$s. This seems to call for a deeper embedding.

Given the encoding of first-order logic and the lambda-calculus, it is not surprising that the system can also encode Higher-Order Logic (omitted due to lack of space). It is worth mentioning that also versions of dependently typed lambda calculi can be represented.

\subsubsection*{Deep Embeddings}
\label{sec:deep-embeddings}
We show an alternative encoding of the pure lambda-calculus, which corresponds to a so-called \emph{deep} embedding, in contrast to the one given at the beginning of this section. The difference has to do fundamentally with the status of \emph{substitution}. Above we have used the atom substitution in the concretion operator of the framework to directly implement the object language substitution, whereas in the approach to be now considered, the latter is given an alpha-recursive characterisation.

First we define the set of names without constructors, i.e., only inhabited by atoms, and
the set of terms:
\[\typing{\name}{\sort} \quad \text{ and }\quad \typing{\Lambda}{\sort}
\]
\noindent{The constructors are defined by:}
\begingroup
\small
\[
\fconstrdecl{\mathsf{Var}}{(\typing{\_}{\name})}{\Lambda}{}\
\quad \fconstrdecl{@}{(\typing{\_}{\Lambda},\typing{\_}{\Lambda})}{\Lambda}{} 
\quad \fconstrdecl{\lambda}{(\typing{\_}{\dabsort{x}{\name}{\Lambda}})}{\Lambda}{}
\]
\endgroup

Induction over lambda-terms is declared as follows:
\[
\begin{aligned}\fconstrdecl{\Lambda_{\mathsf{ind}}}
{(\typing{&P}{\dabsort{t}{\Lambda}{\form}},\\
&\typing{\_}{\dabsort{x}{\name}{\prf{(\concr{P}{\mathsf{Var}(x)})}}},\\
&\typing{\_}{\dabsort{m}{\Lambda}
            {\dabsort{n}{\Lambda}
            {\dabsort{h_1}{\prf(\concr{P}{m})}}
            {\dabsort{h_2}{\prf(\concr{P}{n})}}
            \concr{P}{@(m,n)}
            }},\\
&\typing{\_}{\dabsort{m}{\Lambda}
            {\dabsort{h}{\prf{(\concr{P}{m})}}
            {\prf{(\concr{P}{\lambda(\dabsort{z}{\name}{m})})}}
            }},\\
&\typing{M}{\Lambda}\\
&)}{\prf{(\concr{P}{M}})}{\fresh{z}{P}}
\end{aligned}
\]

Notice the analogy with the principle of alpha-structural induction formulated in~\cite{pitts:alpsri}. In a similar manner,
substitution can be declared as:
\newcommand{\Lamsubst}{\Lambda_{\mathsf{subst}}}
\[\fconstrdecl{\mathsf{\Lamsubst}}{(\typing{\_}{\Lambda},\typing{\_}{\name},\typing{\_}{\Lambda})}{\Lambda}{}\]
and axiomatised by the following equations (where $\Lamsubst(M, x, N)$ stands for  $M[x := N]$):
\begingroup
\small
\[
\begin{aligned}
\fconstrdecl{{\Lamsubst}_{\mathsf{-v}}}{&(\typing{X}{\name},\typing{N}{\Lambda})}
{\prf{(\mathsf{\forall}(
                         \dabsort{x}{\name}
                         {=( 
                             \mathsf{\Lamsubst}(\mathsf{Var}(x), X, N),
                             \mathsf{if}(\mathsf{==}(x, X),N,\mathsf{Var}(x))
                            )}))}}
{}    
\\
\fconstrdecl{{\Lamsubst}_{-@}}{&(\typing{M_{1}}{\Lambda}, \typing{M_{2}}{\Lambda}, \typing{X}{\name},\typing{N}{\Lambda})}
{\prf{(=(\mathsf{\Lamsubst}(@(M_{1}, M_{2}) , X, N),
                   @(\mathsf{\Lamsubst}(M_{1}, X, N), 
                     \mathsf{\Lamsubst}(M_{2}, X, N))   
                  ))}}
{}
\\
\fconstrdecl{{\Lamsubst}_{-\lambda}}{&(\typing{M}{\Lambda}, \typing{X}{\name},\typing{N}{\Lambda})}
{\prf{(=(\mathsf{\Lamsubst}(\lambda(\dabsort{x}{\name}{M}), X, N),
\lambda(\dabsort{x}{\name}{\mathsf{\Lamsubst}(M, X, N))}}}
{\fresh{x}{\{N, X\}}}.
\end{aligned}
\]
\endgroup

Here an operator $\mathsf{==}$ is used to compare names in \name, which yields a boolean in the obvious manner. The set of booleans with the operator $\mathsf{if}$ is introduced as usual.
With these declarations, we have been able to construct a derivation for the substitution lemma

$(M[x := N])[y := P] = (M[y := P])[[x := N[y := P]]$ if $x \not = y$ and $\fresh{x}{P}$,

\noindent which proceeds by (alpha-structural) induction on $M$ in very much the same way as a pencil-and-paper proof utilising Barendregt's variable convention~\cite{BarendregtH:lamcss}.

\section{Conclusions and Related Work}
\label{section:related-work}

One of the best known examples of logical frameworks is LF~\cite{harper:fradl}, based on a typed $\lambda$-calculus with
dependent types. 
Several proof assistants based on the use of Higher-Order Abstract Syntax to encode binders have been implemented (e.g., Beluga~\cite{Beluga}).
Nominal type theory as a basis for logical frameworks has been investigated independently by Cheney~\cite{CheneyJ:sntt,CheneyJ:depntt} and Pitts~\cite{PittsA:deptta} as extensions of a typed $\lambda$-calculus with names, name-abstraction and concretion operators, and name-abstraction types. 
Although the extension with nominal features of a dependently typed $\lambda$-calculus yields a powerful type theory,  the interaction between name abstraction and functional abstraction is a source of difficulties (see~\cite{CheneyJ:sntt} for a detailed discussion). 

We have shown that despite its first-order character, our dependently sorted system can yield a logical framework where standard languages with binders can be defined and reasoned about. For example, we have shown how to define an induction principle for $\lambda$-terms,  taking into account the  $\alpha$-equivalence relation.  
Also, the first-order character of the language, which is a consequence of the restriction for atoms to carry only data sorts, permits a simple definition of computation, actually consisting in simple syntactic definition at the meta-level. This alleviates somewhat the notions and proofs of adequacy, as compared e.g. to \cite{harper:fradl}, just as happens with the lambda-free frameworks \cite{AdamsR:lamfree}.

Cartmell's Generalised Algebraic Theories (GAT)~\cite{Cartmell86,SterlingJ:algttu} also include dependent sorts but lack any intrinsic binding structure. 
To facilitate the specification of languages with binders  second-order versions of GATs have been proposed~\cite{Fiore08,UemuraPhD,KaposiX24}, which incorporate binding and capture-avoiding substitution using free algebras with substitution structure as a model. We adopt a nominal approach: our dependently sorted system can be seen as an extension of GATs with nominal features, such as notions of fresh names and name abstraction, as well as name permutations and capture-avoiding substitutions. 

The sorting system has some limitations that we will address in future work, such as the fact that in declarations for term and type constructors the variables cannot have sorts that depend on atoms, which would be useful to define recursor operators  on $\lambda$-terms. Allowing for variables depending on terms will also permit to use them to 
represent goals to be solved in incomplete terms, as well as schematic derivations.
We also assign great importance to the goal of extending the present framework with recursive definitions to be used as rules of computation, as in Martin-L\"of's type theory. This would conduct us to a version of this theory fully founded upon a nominal syntax, with the issues brought about by binding solved at an infrastructure level.

\subsection*{Acknowledgments} We thank three anonymous referees for their
comments that helped us to improve this final version. This work was partially
funded by Agencia Nacional de Investigación e Innovación (ANII) and SECyT-UNC.

\bibliographystyle{eptcs}
\bibliography{nominal}

\end{document}